\RequirePackage{fix-cm}
\documentclass[smallcondensed, draft, natbib]{svjour3}     
\smartqed  
\usepackage{graphicx}
\usepackage{mathptmx}      
\usepackage{latexsym, amsmath, amssymb}
\newcommand{\btheta}{\boldsymbol{\theta}}
\newcommand{\Vphi}{\boldsymbol{\varphi}}
\newcommand{\bmu}{\boldsymbol{\mu}}
\newcommand{\bEta}{\boldsymbol{\eta}}
\newcommand{\p}{\vec{p}}

\spnewtheorem{result}{Result}{\bf}{\it}
\journalname{}
\begin{document}

\title{On the construction of unbiased estimators for the group testing problem}



\author{Gregory Haber         \and
        Yaakov Malinovsky
}


\institute{Gregory Haber \\
		   Yaakov Malinovsky \at
			  Department of Mathematics and Statistics, \\
			  University of Maryland, Baltimore County \\	
              1000 Hilltop Circle. \\
              Baltimore, MD 21250 \\
              \email{ghaber1@umbc.edu}
              \email{yaakovm@umbc.edu}           
}

\date{Received: date / Accepted: date}

\maketitle

\begin{abstract}
Debiased estimation has long been an area of research in the group testing literature. This has led to the development of several estimators with the goal of bias minimization and, recently, an unbiased estimator based on sequential binomial sampling. Previous research, however, has focused heavily on the simple case where no misclassification is assumed and only one trait is to be tested. In this paper, we consider the problem of unbiased estimation in these broader areas, giving constructions of such estimators for several cases. We show that, outside of the standard case addressed previously in the literature, it is impossible to find any proper unbiased estimator, that is, an estimator giving only values in the parameter space. This is shown to hold generally under any binomial or multinomial sampling plans.
\keywords{Binomial sampling plans \and Group testing \and Multinomial sampling plans \and Sequential estimation \and Unbiased estimation}
\subclass{62F10 \and 62L12}
\end{abstract}

\section{Introduction}
Group testing, which includes generally any situation in which specimens are tested in groups instead of individually, has been an ongoing area of research in the statistical literature for over 70 years. First introduced in \cite{dorfman1943} as a means of screening U.S. Army inductees for syphilis, subsequent research has led to the development of two overarching fields, case identification (as in Dorfman's original work) and estimation.

The estimation problem, which is the focus of the current work, has as its prototypical case the prevalence estimation of a single binary trait from an assumed infinite population when testing is done error free. Typically the trait of interest will be rare, so that grouping can lead to a significant reduction in the number of tests required or an increase in efficiency (in terms of mean square error) for a fixed number of trials.

While the above scenario is important theoretically, in many applications tests will be subject to misclassification error which must be accounted for when analyzing group testing data. Research in this area has been broad, covering an array of cases \citep[see, as a few examples,][]{tu1995, hung1999, liu2012, mcmahan2013, zhang2014, huang2016, li2017}.

An additional area which is becoming increasingly important for applications due to the growing development of multiplex screening tools is the estimation of prevalences for several traits simultaneously. Such assays can be modelled naturally using multinomial sampling and extensions of group testing methods to such designs can be found in \cite{ho2000, pfeiffer2002, tebbs2013, ding2015, warasi2016}.

In all cases, one of the major difficulties in carrying out estimation using group testing is that the standard maximum likelihood estimator (MLE) is biased, often quite significantly depending on the true underlying prevalence and group sizes \citep[see, for example,][]{gibbs1960, thompson1962, swallow1985}. This has led to the development of several alternative debiased estimators, that is, estimators with significantly reduced bias relative to the MLE \citep[see][]{burrows1987, tebbs2003, hepworth2009, ding2016, santos2016}. Perhaps most importantly, such estimators generally yield large reductions in the mean square error (MSE) when compared with the MLE, indicating the importance of developing such tools for group testing data. While the bias and MSE can be controlled to a degree with good design (i.e. appropriate choices of the group sizes), this usually requires the use of prior knowledge regarding the prevalence parameter or adaptive designs which may not be feasible in many cases \citep[see, for example,][]{chiang1962, hos1994, haber2017a}.

It should be noted that, if fixed binomial sampling is used, it is impossible to find any unbiased estimator for the underlying parameter when group testing is used. This fact was mentioned in \cite{hall1963} and follows from a general result concerning the estimation of the function of a binomial parameter given in \cite{lehmann1998}, p. 100. This result can be easily extended to the cases when misclassification is present and/or multiple traits are screened simultaneously, so that an unbiased estimator can not exist in such cases when fixed sampling (either binomial or multinomial, as appropriate) is used.

To consider unbiased estimation for the prevalence in the group testing problem, then, it is necessary to consider the broader class of binomial and multinomial sampling plans (defined in the following section), of which the fixed binomial and multinomial designs are members. In a recent work, \cite{haber2017b} took this approach and showed, based on results from \cite{degroot1959}, that under a certain class of inverse binomial sampling models it is possible to construct an unbiased estimator. Their work, however, was restricted to the simple case outlined above where only a single trait is to be estimated without misclassification.

In this paper, we extend the question of unbiased estimation for group testing to the above generalizations, misclassification and multiple-trait screening. In particular, we focus on the case when misclassification errors are assumed known and on the simultaneous estimation of two correlated diseases.

We show that, in both cases, unbiased estimation is possible using inverse sampling and constructions are provided under the appropriate models. It is shown, however, that these estimators are improper, that is, they lie outsize of the parameter space for some sample values. The core theoretical result of this work is to show that this will be true for any unbiased estimator under any binomial sampling plan with misclassification or any multinomial sampling plan (with at least three elements), even with perfect testing.

\section{Binomial and multinomial sampling plans}
In this section, we define the general classes of binomial and multinomial (of which binomial is a special case) sampling plans. A more detailed treatment of binomial sampling plans can be found in, among others, \cite{girshick1946, degroot1959}. Similar results for multinomial sampling plans can be found in \cite{kremers1990, koike1993}.

In general, a binomial sampling plan $\mathcal{S}$ is a set of points on the non-negative $xy$-plane determined by a set of boundary points $\mathcal{B}_\mathcal{S}$. For all plans, sampling begins at the origin and increases the $x$ or $y$ coordinate with probabilities $\theta$ and $1 - \theta$, respectively, iteratively until a point in $\mathcal{B}_\mathcal{S}$ is reached. This class is very broad, and includes both the fixed binomial and inverse binomial sampling plans, as well as many variations of bounded or fully sequential sampling designs.

The class of multinomial sampling plans is a direct generalization of the above idea. We say $\mathcal{S}_t$ is a multinomial sampling plan in $t + 1$ dimensions if $\mathcal{S}_t$ is a set of points on the non-negative orthant lying in $t + 1$ dimensional space. The plan is similarly determined by a set of boundary points $\mathcal{B}_{\mathcal{S}_t}$. Sampling begins at the origin and increases the $x_i$ coordinate, $i = 0, 1, \ldots t$, at each step with probability $\theta_i, i = 0, 1, \ldots t$ where $\boldsymbol{\theta} = (\theta_1, \theta_2, \ldots, \theta_t)$ is a multinomial parameter and $\theta_0 = 1 - \mathbf{1}^\prime \boldsymbol{\theta}$.

\section{Unbiased estimation under inverse multinomial sampling}
In this section a theorem with necessary and sufficient conditions for unbiased estimation of a function of the parameter vector of an inverse multinomial model is given. This is a generalization of Theorem 4.1 found in \cite{degroot1959}, which applies only for two class problems, and will be used in subsequent sections to construct unbiased estimators under group testing models for one and two traits. While the results presented here are applicable in many situations, for convenience we refer to testing for single or multiple diseases throughout.

Let $\bmu = (\mu_1, \ldots, \mu_t)^\prime$ with $ \mu_0 = 1 - \mathbf{1}^\prime \bmu$ and let $IMN_t(c, \bmu)$ denote the $t$-class inverse multinomial model with parameter $\bmu$ which samples until c observations from the class corresponding to $\mu_0$ are observed . Then, the random variable $\mathbf{X} \sim IMN_t(c, \bmu)$ with parameter space $\overline{\Psi} = \{\bmu: \mathbf{1}^\prime \bmu < 1, 0 \prec \bmu \prec 1 \}$ where $\prec$ denotes element-wise inequality, if
\[P(\mathbf{X} = \mathbf{x}) = \binom{c + \sum_{i=1}^t x_i - 1}{c-1, x_1, \ldots, x_t} \mu_0^c \prod_{i=1}^t \mu_i^{x_i}.\] Let $\Psi \subset \overline{\Psi}$ with $int(\Psi)$ the interior of $\Psi$. Then, we say that $X$ has an inverse multinomial distribution with restricted parameter space if $X$ has the same pdf as above but with parameter space $\Psi$. A special case of this distribution with $t=1$ is the inverse binomial, which corresponds to the classical group testing problem when screening for one disease.

\begin{theorem}
	\label{mnub}
	Let $\mathbf{X} \sim IMN_t(c, \boldsymbol{\mu})$ with restricted parameter space $\Psi$. A function $h(\boldsymbol{\mu})$ is estimable unbiasedly for all $\bmu \in int(\Psi)$ if and only if $h$ is an analytic function on a region containing an open-ball about $\mathbf{0} \in \mathbb{R}^t$ and $int(\Psi)$. The estimator is given by
	\[f(\mathbf{x}) = \frac{(c-1)!}{(c + \sum_{i=1}^t x_i - 1)!} \left.\frac{\partial^{\sum_{i=1}^t x_i }g(\boldsymbol{\mu})}{\partial \mu_1^{x_1} \cdots \partial \mu_t^{x_t}}\right|_{\boldsymbol{\mu} = \mathbf{0}},\] where $g(\boldsymbol{\mu}) = \frac{\displaystyle{h(\boldsymbol{\mu})}}{\displaystyle{\mu_0^c}}.$
\end{theorem}

\begin{remark}
	It should be noted that the restriction to values in $int(\Psi)$ is not absolute, and values from the boundary such as $\mu_0 = 1$ may be estimable unbiasedly, while others such as values in the plane $\mu_0 = 0$ can not. The latter point can be seen by noting that the function $g(\boldsymbol{\mu})$ above is undefined when $\mu_0 = 0$. In general, there is little interest in estimating any points on the boundary, hence the restriction to the interior is sufficient.
\end{remark}

While we will be interested here in applying this theorem only in the context of group testing, other possibilities include the unbiased estimation of  the relative risk or odds ratio of two diseases estimated simultaneously.

\section{One disease case with misclassification}
For the single disease group testing problem, we assume an infinite population of individuals whose binary status can be represented by independent random variables $\varphi \sim Ber(p)$. In what follows, $p$ is the quantity we seek to estimate. If, instead of as individuals, members of this population are tested in groups of size $k$, we have the new random variable
$\vartheta = \max\{\varphi_{1}, \ldots, \varphi_{k}\} \sim Ber(1 - q^k)$, where $q = 1 - p$.

To incorporate testing error, let $\tilde{\vartheta}$ be the true, latent value of the observed $\vartheta$. Then, we define the specificity and sensitivity, respectively, as $\pi_0 = P(\vartheta = 0|\tilde{\vartheta} = 0)$ and $\pi_1 = P(\vartheta = 1|\tilde{\vartheta} = 1)$. This yields the distribution $\vartheta \sim Ber(\theta)$ where $\theta = \pi_1 - \nu q^k$ with $\nu = \pi_1 + \pi_0 - 1$. It should be noted that this model is identifiable if and only if $\nu \neq 0$. A standard assumption to address this, which is made here as well, is that both $\pi_0$ and $\pi_1$ are greater than $0.5$. This is reasonable as it merely assumes the test performs better than random guessing.

To find an unbiased estimator using Theorem 1, we consider $Y \sim IMN_1(c, \theta)$ which is the number of positives until $c$ groups testing negative for the disease are observed. Note that the parameter space here is restricted as a function of $\theta$ when either $\pi_0 < 1$ or $\pi_1 < 1$ since, for $0 < p < 1$, $1 - \pi_0 < \theta < \pi_1$. Then, we seek an unbiased estimator of $\displaystyle{q = h(\theta) = \frac{(\pi_1 - \theta)^{1 / k}}{\nu^{1 / k}}}$, which is analytic on the interval $|\theta| < \pi_1$.

\begin{result}
	\label{res:ub1}
	For the one disease case with misclassification, with $Y \sim IMN_1(c, \pi_1 - \nu q^k)$, the unique unbiased estimator of $p$ is given by
	\[\hat{p}_{UB}(y) =1 -\left(\frac{\pi_1}{\nu}\right)^{1 / k}\sum_{i = 0}^{y} \binom{y}{i}\frac{1}{\pi_1^{y - i}}\frac{(c + i - 1)!}{(c + y - 1)!} \prod_{j = i}^{y}\frac{(y - j - 1 / k)}{(y - i -  1 / k)},y = 0, 1, 2, \ldots.\]
\end{result}

This result can be used to derive an unbiased estimator for the perfect testing case $(\pi_0 = \pi_1 = 1)$ as given in the following corollary and \cite{haber2017b}.

\begin{corollary}
	\label{cor:ub1}
	For the one disease case with no misclassification, an unbiased estimator of $p$ is given by
	\[\hat{p}_{UB}(y) = 1 - \frac{1}{\left(1 - \frac{1}{k(c + y)}\right)}\prod_{i=0}^y \left(1 - \frac{1}{k(c + i)}\right),  y = 0, 1, 2, \ldots\]
\end{corollary}

\subsection{Non-properness of unbiased estimator}
While the estimator given in Result \ref{res:ub1} is unbiased, for either $\pi_0 < 1$ or $\pi_1 < 1$ it is an improper estimator, that is it yields values lying outside the parameter space.

To see this, note that if $\pi_0 < 1$ then, for any $\pi_1$, we have $\hat{p}_{UB}(0) = 1 - \left(\frac{\pi_ 1}{\pi_1 + \pi_0  - 1}\right)^{1 / k} < 0$. Likewise, If $\pi_0 = 1$, then, for $\pi_1 < 1$ and $y \geq 1$,
\[\hat{p}_{UB}(y) = \frac{1}{k}\sum_{i = 0}^{y - 1} \binom{y}{i}\frac{1}{\pi_1^{y - i}}  \frac{(c + i - 1)!}{(c + y - 1)!}\prod_{j = i}^{y - 1}\frac{(y - j - 1 / k)}{(y - i - 1 / k)}.\] Now, each term of the sum in this expression is positive, so it is sufficient to show that for some $y$ at least one term is greater  than $1$, resulting in a total estimate larger than $1$. For the $i = 0$ term, we have $\displaystyle{\frac{1}{kc}\prod_{j = 1}^{y - 1} \frac{1}{\pi_1}\left(1 - \frac{c + 1 / k}{c + y - j}\right)}$ which diverges since $\pi_1 < 1$.

While these results, combined with the necessity clause of Theorem \ref{mnub}, mean that there exist no proper unbiased estimators under the inverse binomial model, in the following result we extend this idea to show that no such estimator exists under any binomial sampling plan when misclassification is present.

\begin{theorem}
	\label{thm:np1}
	Let $\mathcal{S}$ be a binomial sampling plan with set of boundary points $\mathcal{B}_\mathcal{S}$ for which, at a given step, the $x$ and $y$ coordinates are increased with probability $\theta = \pi_1 - \nu q^k$ and $1 - \theta$, respectively. Then, if $\pi_0 < 1$ or $\pi_1 < 1$, there exists no proper unbiased estimator of $p$ under $\mathcal{S}$.
\end{theorem}

It should be noted that, while the explicit construction of the unbiased estimator in Result \ref{res:ub1} required the assumption that the misclassification parameters were known, the result of Theorem \ref{thm:np1} holds more generally, even when this assumption does not hold.

From the proof of Theorem \ref{thm:np1} we get the following corollary, which is also given in \cite{haber2017b}, showing that the above inverse binomial model which counts until $c$ negatives is the only one yielding an unbiased estimator of $p$.

\begin{corollary}
	\label{cor:np1}
	Let $Y\sim IMN_1(c, 1-\theta)$ where $\theta = \pi_1 - \nu q^k$, so that $Y$ is the number of negative groups drawn until $c$ positive results are observed. Then, there exists no unbiased estimator of $p$ for any values of $\pi_0$ and $\pi_1$.
\end{corollary}

\section{Two disease case with no misclassification}
For the case of two diseases, let $\varphi_{1}$ and $\varphi_{2}$ be marginally binomial random variables with parameters $0 < p_1 < 1$ and $0 < p_2 < 1$ respectively. Then, $(\varphi_{1}, \varphi_{2})$ has a one-to-one correspondence to the vector $\Vphi = (\varphi_{00}, \varphi_{10}, \varphi_{01}, \varphi_{11})$ with joint multinomial distribution $\Vphi \sim MN_3(1, \mathbf{p})$ and sample space $\boldsymbol{\Psi}_{\mathbf{p}} = \{\mathbf{p}: \mathbf{1}^\prime \mathbf{p} < 1, 0 \prec \mathbf{p} \prec 1\}$, where $\mathbf{p} = (p_{10}, p_{01}, p_{11})$ and $p_{00} = 1 - \mathbf{1}^\prime \mathbf{p}$. Note that the marginal parameters can be expressed as $p_1 = p_{10} + p_{11}$ and $p_2 = p_{01} + p_{11}$.

If we assume no misclassification, we have the $i$th grouped sample $(\vartheta^{(k)}_{1i}, \vartheta^{(k)}_{2i}) =\\ (\max\{\varphi_{1i_1}, \cdots, \varphi_{1i_k}\}, \max\{\varphi_{2i_1}, \cdots, \varphi_{2i_k}\})$ which corresponds to
\[\boldsymbol{\vartheta}_i^{(k)} = (\vartheta_{00}, \vartheta_{10}, \vartheta_{01}, \vartheta_{11}) \sim MN_3(1, \btheta),\]
where
\begin{align}
\label{eq:theta1}
\btheta &= (\theta_{10}, \theta_{01}, \theta_{11}) \nonumber \\
&= ((p_{00} + p_{10})^k - p_{00}^k, (p_{00} + p_{01})^k - p_{00}^k,1 - (p_{00} + p_{10})^k - (p_{00} + p_{01})^k + p_{00}^k)
\end{align} and
\begin{equation}
\label{eq:theta2}
\theta_{00} = 1 - \mathbf{1}^\prime \btheta = p_{00}^k.
\end{equation}

If we sample until $c$ groups are found without either disease, and set $\mathbf{Z} = (z_{10}, z_{01}, z_{11})$ to be the sum of the observed $\boldsymbol{\vartheta}^{(k)}_i$s, we have $\mathbf{Z} \sim IMN_3(c, \btheta)$. Note that the parameter space of $\mathbf{Z}$, $\boldsymbol{\Psi}_Z = \{\btheta(\mathbf{p}): \mathbf{1}^\prime \mathbf{p} < 1, 0 \prec \mathbf{p} \prec 1\}$, is a proper subset of the full parameter space $\boldsymbol{\Psi}_{\btheta} = \{\btheta: \mathbf{1}^\prime \btheta < 1, 0 \prec \btheta \prec 1\}$. This fact will play a crucial role below in showing that there exists no proper unbiased estimator of $\mathbf{p}$.

To find an unbiased estimator, the following lemma  will be needed, which is simply the result of inverting (\ref{eq:theta1}) and (\ref{eq:theta2}).

\begin{lemma}
	\label{lm:h2}
	The unique function $h:\btheta \mapsto \mathbf{p}$ is given by
	\begin{align*}
    p_{00} &= h_{00}(\btheta) = (1 - \theta_{10} - \theta_{01} -
	\theta_{11})^{1 / k},\\
	p_{10} &= h_{10}(\btheta) = (1 - \theta_{01} - \theta_{11})^{1 / k} -
	h_{00}(\btheta),\\
	p_{01} &= h_{01}(\btheta) = (1 - \theta_{10} - \theta_{11})^{1 / k} -
	h_{00}(\btheta),\\
	p_{11} &= h_{11}(\btheta) = 1 - p_{00} - p_{10} - p_{01}.	
	\end{align*}
\end{lemma}

The function $h(\btheta)$ given in Lemma \ref{lm:h2} is analytic on an open region containing $\mathbf{0}\cup int(\boldsymbol{\Psi}_Z)$, so the conditions of Theorem \ref{mnub} hold and an unbiased estimator exists.
\begin{result}
	\label{res:ub2}
	The unique unbiased estimator of $\mathbf{p}$ where $\mathbf{Z} \sim IMN_3(c, \btheta)$, with $\btheta$ as in (\ref{eq:theta1}), is given by
	\begin{align*}
		\hat{p}_{00} &= \frac{1}{\left(1 - \frac{1}{k(c + z_{10} + z_{01} + z_{11})}\right)}\prod_{j=0}^{z_{10} + z_{01} + z_{11}}\left(1 - \frac{1}{k(c + j)}\right),\\
		\hat{p}_{10} &= \frac{1}{\left(1 - \frac{1}{k(c + z_{10} + z_{01} + z_{11})}\right)}\prod_{j=0}^{z_{01} + z_{11}}\left(1 - \frac{1}{k(c + z_{10} + j)}\right) - \hat{p}_{00},\\
		\hat{p}_{01} &= \frac{1}{\left(1 - \frac{1}{k(c + z_{10} + z_{01} + z_{11})}\right)}\prod_{j=0}^{z_{10} + z_{11}}\left(1 - \frac{1}{k(c + z_{01} + j)}\right) - \hat{p}_{00},\\
		\hat{p}_{11} &= 1 - \hat{p}_{00} - \hat{p}_{10} - \hat{p}_{01}.
	\end{align*}
\end{result}

The unbiased estimator given in Result \ref{res:ub2} is an improper estimator.
This can be shown by counterexample, considering the point $\mathbf{z} = (1, 1, 0)$. We have, evaluating at this point,
\begin{align*}
\hat{p}_{00} + \hat{p}_{10} + \hat{p}_{01} &= 2\left(1 - \frac{1}{k(c + 1)}\right) - \left(1 - \frac{1}{kc}\right)\left(1 - \frac{1}{k(c + 1)}\right)\\
&=1 + \frac{1}{kc} - \frac{1}{k(c+1)} -\frac{1}{k^2c(c + 1)}\\
&= 1 + \frac{1}{kc}\left(1 - \frac{(c + 1 / k)}{c + 1}\right) \\
& > 1,
\end{align*}
for any $c$ and $k > 1$.

As in Theorem \ref{thm:np1}, this property can be shown to hold for any unbiased estimator under any multinomial sampling plan.

\begin{theorem}
	\label{thm:np2}
	Let $\mathcal{S}_3$ be a multinomial sampling plan in four dimensions with boundary points $\mathcal{B}_{\mathcal{S}_3}$ such that at each step the $i$th coordinate, $i = 0, 1, 2, 3$, is increased with probability $\theta_i$ as in (\ref{eq:theta1}) and (\ref{eq:theta2}). Then, there exists no proper unbiased estimator of $\mathbf{p}$ under $\mathcal{S}_3$.
\end{theorem}

\section{Two disease case with misclassification}
In this section we consider the two disease testing problem when misclassification is present, by looking at two models for incorporating such testing errors. In both cases we will assume the misclassification parameters to be known a priori, although the results on non-proper estimators will hold more generally.

The first model, as introduced in \cite{li2017}, is very general, requiring no assumptions on how the marginal testing errors are combined. The downside, as we shall see, is that this requires a large number of parameters, knowledge of which may not be available for many of the assays used in applications. Let $\tilde{\varphi}_{a},\ a \in \{00, 10, 01, 11\}$ be the true latent value of the observed random variable $\varphi_a$. Then, we have the misclassification parameters $\pi_{a|b} = P(\varphi_a|\tilde{\varphi}_b),\ a, b \in \{00, 10, 01, 11\}$. While this indicates 16 parameters, each one can be expressed as a linear combination of three others, so that the model consists of twelve extra parameters.

If we again let $\mathbf{Z}$ be the sum of the $\boldsymbol{\vartheta}_i^{(k)}$s until $c$ groups are observed without either disease then, with the above misclassification values, and $\btheta$ as in (\ref{eq:theta1}) and (\ref{eq:theta2}), we now have $\mathbf{Z} \sim IMN_3(c, \bEta)$ where $\bEta = (\eta_{10}, \eta_{01}, \eta_{11})$ and $\eta_{00} = 1 - \mathbf{1}^\prime \bEta$ with
\[\eta_a = \pi_{a|00}\theta_{00} + \pi_{a|10}\theta_{10} + \pi_{a|01}\theta_{01} + \pi_{a|11}\theta_{11},\ a \in \{00, 10, 01, 11\}.\]

With $\boldsymbol{\pi}_{00} = (\pi_{10|00}, \pi_{01|00}, \pi_{11|00})^\prime$, and
\[\arraycolsep=5pt
\boldsymbol{\Phi} = \left(\begin{array}{ccc}
\pi_{10|10} - \pi_{10| 00} & \pi_{10|01} - \pi_{10| 00} & \pi_{10|11} - \pi_{10| 00} \\
\pi_{01|10} - \pi_{01| 00} & \pi_{01|01} - \pi_{01| 00} & \pi_{01|11} - \pi_{01| 00} \\
\pi_{11|10} - \pi_{11| 00} & \pi_{11|01} - \pi_{11| 00} & \pi_{11|11} - \pi_{11|00} \end{array}\right),\] the parameter vector for this model can be expressed succinctly as
\begin{equation}
\label{eq:eta1}
\bEta = \boldsymbol{\pi}_{00} + \boldsymbol{\Phi}\btheta.
\end{equation}

\subsection{Independent misclassification errors}
\label{sec:inde}
An alternative, simplified, model to the above assumes there are only four misclassification parameters, specificity and sensitivity for each marginal disease, and that the joint errors can be found assuming independence. Examples of this model can be found in \cite{pfeiffer2002} and \cite{tebbs2013}, among others.   Formally, if $\pi_0^{(i)}$ and $\pi_{1}^{(i)}$ are the specificity and sensitivity, respectively, for the test under the $i$th disease, $i = 1, 2$, then we assume $\pi_{10|00} = (1 - \pi_{0}^{(1)})\pi_{0}^{(2)}$, $\pi_{10|10} = \pi_{1}^{(1)}\pi_{0}^{(2)}$, and so on for all twelve parameters above.

\subsection{Identifiability of model}

Before addressing the question of unbiased estimation, we first consider conditions to ensure the models presented above are identifiable. This is an important question which has yet to be dealt with explicitly in the literature.

\begin{theorem}
	\label{thm:ident}
	Let $\mathbf{Z} \sim IMN_3(c, \bEta)$ with $\bEta$ as in (\ref{eq:eta1}). Then, the model is identifiable if and only if the determinant of $\boldsymbol{\Phi}$ is non-zero, that is
	\[\arraycolsep=5pt
	|\boldsymbol{\Phi}| = \left|\begin{array}{ccc}
	\pi_{10|10} - \pi_{10| 00} & \pi_{10|01} - \pi_{10| 00} & \pi_{10|11} - \pi_{10| 00} \\
	\pi_{01|10} - \pi_{01| 00} & \pi_{01|01} - \pi_{01| 00} & \pi_{01|11} - \pi_{01| 00} \\
	\pi_{11|10} - \pi_{11| 00} & \pi_{11|01} - \pi_{11| 00} & \pi_{11|11} - \pi_{11|00} \end{array}\right| \neq 0.\]
\end{theorem}

\begin{corollary}
	\label{cor:ident}
	For the independent errors model given in Section \ref{sec:inde}, the model is identifiable if and only if both $\pi_0^{(1)} + \pi_1^{(1)} \neq 1$ and $\pi_0^{(2)} + \pi_{1}^{(2)} \neq 1$.
\end{corollary}

Similar to the one-disease case, the conditions of Corollary \ref{cor:ident} will always be satisfied if we make the reasonable assumption that all misclassification parameters are greater than $0.5$. The more general case, as presented in Theorem \ref{thm:ident}, is easy to check in a given situation, but does not easily yield itself to simplified conditions.

\subsection{Non-properness of unbiased estimator for two diseases with misclassification}

As in the case with no misclassification, Theorem 1 can be used to construct an unbiased estimator under either of the misclassification models presented above. Since this construction is merely a technical generalization of the previous two cases, it is excluded here.

Likewise, as in the previous cases, we can generalize Theorem \ref{thm:np2} to show that this holds under any multinomial sampling plan when misclassification is present. The proof of this result follows directly from Theorem \ref{thm:np2} since, assuming the conditions of Theorem \ref{thm:ident} hold, $\bEta$ is a full rank affine transformation of $\btheta$.

\begin{theorem}
	\label{thm:np2mc}
	Let $\mathcal{S}_3$ be a multinomial sampling plan in four dimensions with boundary points $\mathcal{B}_{\mathcal{S}_3}$ such that at each step the $i$th coordinate, $i = 0, 1, 2, 3$, is increased with probability $\eta_i$, where $\bEta$ is as given in (\ref{eq:eta1}). Then, there exists no proper unbiased estimator of $\mathbf{p}$ under $\mathcal{S}_3$.
\end{theorem}

\section{Discussion}
We have shown that, outside of the standard case when testing one disease with no misclassification, it is impossible to get a proper unbiased estimator in the group testing problem. This result holds very generally, under any design from the classes of binomial or multinomial sampling plans, not only those previously considered in the literature. While, for the multinomial case, we have provided proofs only for two diseases, the same techniques can be applied to show this result holds for any number of traits.

Of course, such scenarios are the norm in applications, so the question remains as to how estimation should best be carried out in light of the present bias.
For one disease with misclassification, there do exist limited results for this problem. For example, under fixed binomial sampling, the first order bias correction given in \cite{tu1995} can be used to construct a debiased estimator. Still, much more work is needed in this area, analogous to the wide array of estimators in the literature for one disease when no misclassification is assumed. One possible approach is to construct minimal bias estimators as described in \cite{sirazdinov1956} and \cite{hall1963}, although approaches minimizing the risk are generally more favored in the statistical literature. Alternatively, approaches such as those found in \cite{bilder2005} and \cite{hepworth2017}, among others, could possibly be extended to include misclassification. More work is needed to understand how such approaches might generalize, and what the properties of the resultant estimators will be.

For the two disease case, unfortunately, it is much more difficult to give recommendations at this time. The problem in this case is much harder since it is both multivariate and has a restricted parameter space, even without misclassification. There currently exist no results in the literature related to bias reduction for the the two disease group testing scenario. This will be an important area for future research if group testing methods are to be applied in such cases.

\section{Proofs}
\subsection{Proof of Theorem \ref{mnub}}
Since $h$ is analytic, both $h$ and $g(\boldsymbol{\mu}) = \frac{h(\boldsymbol{\mu})}{\mu_0^c}$ can be expanded as a Taylor series over an appropriate region, say $R$. This expansion has the form,
\[g(\boldsymbol{\mu}) = \sum_{x_1, \ldots, x_t = 0}^\infty \frac{1}{x_1!\cdots x_t!}\left. \frac{\partial^{\sum_{i=1}^t x_i }g(\boldsymbol{\mu})}{\partial \mu_1^{x_1} \cdots \partial \mu_t^{x_t}}\right|_{\boldsymbol{\mu} = \mathbf{0}} \prod_{i=1}^t \mu_i^{x_i}.\]

Then, we have
\begin{align*}
\mathrm{E}\left(f(\mathbf{X})\right) &= \sum_{x_1, \ldots, x_t = 0}^\infty f(\mathbf{x}) \binom{c + \sum_{i=1}^t x_i - 1}{c-1, x_1, \ldots, x_t} \mu_0^c \prod_{i=1}^t \mu_i^{x_i} \\
&= \mu_0^c\sum_{x_1, \ldots, x_t = 0}^\infty \frac{1}{x_1!\cdots x_t!} \left.\frac{\partial^{\sum_{i=1}^t x_i }g(\boldsymbol{\mu})}{\partial \mu_1^{x_1} \cdots \partial \mu_t^{x_t}}\right|_{\boldsymbol{\mu} = \mathbf{0}} \prod_{i=1}^t \mu_i^{x_i} \\
&= \mu_0^c g(\boldsymbol{\mu}) \\
&= h(\boldsymbol{\mu}),\end{align*}
for all $\bmu \in int(\Psi)$.

Conversely, if $h$ is estimable unbiasedly, we have for a function $\delta(\mathbf{x})$ and any $\bmu \in int(M)$
\[h(\boldsymbol{\mu}) = \sum_{x_1, \ldots, x_t = 0}^\infty \delta(\mathbf{x}) \binom{c + \sum_{i=1}^t x_i - 1}{c-1, x_1, \ldots, x_t} \mu_0^c \prod_{i=1}^t \mu_i^{x_i}.\] Since this holds for any $\bmu \in int(\Psi)$, $h$ is an analytic function on $R$, hence has a unique Taylor expansion. It follows then that
\begin{align*}
g(\boldsymbol{\mu}) &= \sum_{x_1, \ldots, x_t = 0}^\infty \frac{1}{x_1!\cdots x_t!}\left. \frac{\partial^{\sum_{i=1}^t x_i }g(\boldsymbol{\mu})}{\partial \mu_1^{x_1} \cdots \partial \mu_t^{x_t}}\right|_{\boldsymbol{\mu} = \mathbf{0}} \prod_{i=1}^t \mu_i^{x_i}\\ &= \sum_{x_1, \ldots, x_t = 0}^\infty \delta(\mathbf{x}) \binom{c + \sum_{i=1}^t x_i - 1}{c-1, x_1, \ldots, x_t} \prod_{i=1}^t \mu_i^{x_i},
\end{align*} and equating terms yields
\[\delta(\mathbf{x}) = \frac{(c-1)!}{(c + \sum_{i=1}^t x_i - 1)!} \left.\frac{\partial^{\sum_{i=1}^t x_i }g(\boldsymbol{\mu})}{\partial \mu_1^{x_1} \cdots \partial \mu_t^{x_t}}\right|_{\boldsymbol{\mu} = \mathbf{0}} = f(\mathbf{x})\]
for each $\mathbf{x}$.

\subsection{Proof of Result \ref{res:ub1}}
To apply Theorem 1 in this case, we will require the following lemma.
\begin{lemma}
	Let $\displaystyle{g(\theta) = \frac{(\pi_1 - \theta)^{\xi}}{(1 - \theta)^c}}$, where $\xi = 1 / k$. Then, for any non-negative integer $t$, \[\frac{d^t g(\theta)}{d \theta^t} = \sum_{i=0}^t \binom{t}{i} \frac{(\pi_1 - \theta)^{\xi + i - t}}{(1 - \theta)^{c + i}} \frac{(c + i -1)!}{(c - 1)!}\prod_{j=i}^t \frac{(t - j - \xi)}{(t - i - \xi)}.\]
\end{lemma}

\begin{proof}
		For $t=0$ and $t=1$ the result is a straightforward calculation, and we prove the general case using induction. Suppose the statement holds for $t = m$ so that
		\[\frac{d^{m + 1} g(\theta)}{d \theta^{m + 1}} = \sum_{i=0}^m \binom{m}{i} \frac{(c + i -1)!}{(c - 1)!}\prod_{j=i}^{m} \frac{(m - j - \xi)}{(m - i - \xi)}\frac{d}{d\theta}\left(\frac{(\pi_1 - \theta)^{\xi + i - m}}{(1 - \theta)^{c + i}}\right) ,\] and, for each $i$,
		\[\frac{d}{d\theta}\left(\frac{(\pi_1 - \theta)^{\xi + i - m}}{(1 - \theta)^{c + i}}\right) = \frac{(m - i - \xi)(\pi_i - \theta)^{\xi + i - m - 1}}{(1 - \theta)^{c + i}} + \frac{(c + i)(\pi_1 - \theta)^{\xi + i - m}}{(1 - \theta)^{c + i + 1}}.\]
		
		We now look at the resultant coefficients of the terms $\displaystyle{\frac{(\pi_1 - \theta)^{\xi + i- (m + 1)}}{(1 - \theta)^{c + i}}}$ for each $i$.
		
		For $i=0$ and $i= m$ we have, respectively, $\displaystyle{\prod_{j=0}^{m + 1} \frac{(m + 1 - j - \xi)}{(m + 1 - \xi)}}$ and $\displaystyle{\frac{(c + m + 1 - 1)!}{(c - 1)!}}$.
		
		For $1 \leq i \leq m - 1$ we have
		\begin{align*}
		&\binom{m}{i}\frac{(c + i - 1)!}{(c - 1)!}\prod_{j=i}^m \frac{(m - j - \xi)}{(m - i - \xi)} (m - i - \xi)\\ &\qquad + \binom{m}{i-1} \frac{(c + i - 2)!}{(c - 1)!}\prod_{j=i-1}^m \frac{(m - j - \xi)}{(m - i + 1 - \xi)} (c + i - 1)  \\
		&=\left[\binom{m}{i} + \binom{m}{i-1}\right]\frac{(c + i - 1)!}{(c - 1)!}\prod_{j=i}^{m + 1} \frac{(m + 1 - j - \xi)}{(m + 1 - i - \xi)}\\
		&=\binom{m+1}{i}\frac{(c + i - 1)!}{(c - 1)!}\prod_{j=i}^{m + 1} \frac{(m + 1 - j - \xi)}{(m + 1 - i - \xi)}
		\end{align*}
		
		Combining yields
		\[\frac{d^{m + 1} g(\theta)}{d \theta^{m + 1}} = \sum_{i=0}^{m + 1} \binom{m + 1}{i} \frac{(\pi_1 - \theta)^{\xi + i - (m + 1)}}{(1 - \theta)^{c + i}} \frac{(c + i -1)!}{(c - 1)!}\prod_{j=i}^{m + 1} \frac{(m+1 - j - \xi)}{(m + 1 - i - \xi)}.\]
\end{proof}

Now, to apply Theorem 1, we have $\displaystyle{g(\theta) = \frac{(\pi_1 - \theta)^{\xi}}{\nu^{\xi}(1 - \theta)^c}}$, with $\xi = 1 / k$, which, by Lemma 2 has $t$th derivative
\[\frac{d^t g(\theta)}{d \theta^t} =  \frac{1}{\nu^{\xi}}\sum_{i=0}^t \binom{t}{i} \frac{(\pi_1 - \theta)^{\xi + i - t}}{(1 - \theta)^{c + i}} \frac{(c + i -1)!}{(c - 1)!}\prod_{j=i}^t \frac{(t - j - \xi)}{(t - i - \xi)}.\] Evaluating at $\theta = 0$ yields
\[\left. \frac{d^t g(\theta)}{d \theta^t}\right|_{\theta = 0} =  \left(\frac{\pi_1}{\nu}\right)^{\xi}\sum_{i=0}^t \binom{t}{i} \frac{1}{\pi_1^{t - i}}\frac{(c + i -1)!}{(c - 1)!}\prod_{j=i}^t \frac{(t - j - \xi)}{(t - i - \xi)}, y = 0, 1, 2, \ldots\]

Then, direct application of Theorem 1 yields
\begin{align*}
\hat{q}_{UB}(y) &= \frac{(c-1)!}{(c + y - 1)!} \left(\frac{\pi_1}{\nu}\right)^{\xi}\sum_{i=0}^y \binom{y}{i} \frac{1}{\pi_1^{y - i}}\frac{(c + i -1)!}{(c - 1)!}\prod_{j=i}^y \frac{(y - j - \xi)}{(y - i - \xi)}\\
&= \left(\frac{\pi_1}{\nu}\right)^{\xi}\sum_{i=0}^y \binom{y}{i} \frac{1}{\pi_1^{y - i}}\frac{(c + i -1)!}{(c + y- 1)!}\prod_{j=i}^y \frac{(y - j - \xi)}{(y - i - \xi)}.\\
\end{align*}
Subtracting the above from one gives the desired unbiased estimator of $p$.

\subsection{Proof of Corollary \ref{cor:ub1}}
While it is possible to derive this result algebraically from Result 1, we provide here a much simpler direct proof using Theorem \ref{mnub}.
We have $q = h(\theta) = ( 1 - \theta)^{\xi}$, where $\xi = 1 / k$, so we set $\displaystyle{g(\theta) = \frac{(1 - \theta)^{\xi}}{(1 - \theta)^c} = (1 - \theta)^{\xi - c}}$. Differentiating $t$ times with respect to $\theta$ yields
\begin{align*}
g^{(t)}(\theta) &= (-1)^t(\xi - c)(\xi - c - 1) \times \cdots \times (\xi - c - t + 1) (1 - \theta)^{\xi - c - t}\\ &= \frac{1}{(c + t - \xi)}\prod_{i = 0}^t(c  + i - \xi) (1 - \theta)^{\xi - c - m}, t = 0, 1, 2, \ldots
\end{align*} Evaluating this derivative at $\theta = 0$ and applying Theorem \ref{mnub} yields
\begin{align*}
\hat{q}_{UB}(y) &=\frac{(c-1)!}{(c + y - 1)!}\frac{1}{(c + y - \xi)}\prod_{i = 0}^y(c  + i - \xi)  \\
&= \frac{1}{\left(1 - \frac{1}{k(c + y)}\right)}\prod_{i = 0}^y \left(1 - \frac{1}{k(c + i)}\right), y = 0, 1, 2, \ldots
\end{align*}
As above, the unbiased estimator of $p$ is then found by subtracting this value from $1$.

\subsection{Proof of Theorem \ref{thm:np1}}
For each $(x, y) \in \mathcal{B}$, let $K(x, y)$ be the number of ways to reach the given point, and suppose that $f(x, y)$ is an unbiased estimator of $h(\theta) = q$. Then, we have
\begin{equation}
\label{eq:pfnp1_1}
h(\theta) = \sum_{i=0}^{\infty}\frac{1}{i!}\left.\frac{\partial^i h(\theta)}{\partial \theta^i}\right|_{\theta = 0}\theta^i = \sum_{(x, y) \in \mathcal{B}} f(x, y)K(x, y) \theta^x (1 - \theta)^y \ \text{for all}\ \theta.\end{equation} Since the coefficients for each power of $\theta$ on both sides of the equality must be the same, there must exist a point $(0, y^*) \in \mathcal{B}$ such that $f(0, y^*)K(0, y^*) = h(0)$. Now, there is at most one path to any point on the $y$-axis so, since $(0, y^*) \in \mathcal{B}$, we have $K(0, y^*) = 1$. This yields, $\displaystyle{f(0, y^*) = h(0) = \left(\frac{\pi_ 1}{\pi_1 + \pi_0  - 1}\right)^{1 / k} > 1}$ whenever $\pi_0 < 1$.

Suppose now that $\pi_0 = 1$ and $\pi_1 < 1$. From the above argument, the term on the right hand side of (\ref{eq:pfnp1_1}) associated with the point $(0, y^*)$ reduces to $(1 - \theta)^{y^*}$, so that
\[\sum_{(x, y) \in \mathcal{B} \backslash \{(0,y^*)\}} f(x, y)K(x, y) \theta^x (1 - \theta)^y = q - (1 - \theta)^{y^*}\ \text{for all}\ \theta.\] Allowing $q \to 0$, which is equivalent to $\theta \to \pi_1$, we have
\[\sum_{(x, y) \in \mathcal{B} \backslash \{(0,y^*)\}} f(x, y)K(x, y) \theta^x (1 - \theta)^y \to -(1 - \pi_1)^{y^*} < 0,\] which implies $f(x,y) < 0$ for at least one point.

\subsection{Proof of Corollary \ref{cor:np1}}
From (\ref{eq:pfnp1_1}) in the proof of Theorem \ref{thm:np1}, we see that any sampling plan yielding an unbiased estimator must have exactly one point on the $y$ axis among its boundary points. If $Y\sim IMN_1(c, 1 - \btheta)$ is the number of negatives until $c$ positives are observed, however, then sampling stops if and only if a point on the line $x=c$ is reached. This implies that there is no stopping point along the $y$ axis for the random variable $Y$, hence no unbiased estimator can exist.

\subsection{Proof of Result \ref{res:ub2}}
As in Result \ref{res:ub1}, to apply Theorem \ref{mnub} we require the following lemma giving derivatives of the function $g(\btheta)$.

\begin{lemma}
	\label{lm:ub2}
	Let $\displaystyle{g(\boldsymbol{\theta}) = \frac{h\left(\boldsymbol{\theta}\right)}{\theta_{00}^c}}$. Then, for non-negative integers $z_{10}, z_{01}, z_{11}$ and $\btheta \in \boldsymbol{\Psi}_Z$,
	\begin{description}
		\item[(i)]$\displaystyle{\frac{\partial^{z_{10} + z_{01} + z_{11}} g_{00}(\btheta)}{\partial \theta_{10}^{z_{10}} \partial \theta_{01}^{z_{01}} \partial \theta_{11}^{z_{11}}} = \prod_{j=0}^{z_{10} + z_{01} + z_{11}}\frac{(c + j - 1 / k)}{(c + z_{10} + z_{01} + z_{11} - 1 / k)}(1 - \mathbf{1}^\prime \btheta)^{1 / k -c -  z_{10} - z_{01} - z_{11}};}$
		\vspace{4pt}
		\item[(ii)]
		$\begin{aligned}[t]
		\frac{\partial^{z_{10} + z_{01} + z_{11}} g_{10}(\btheta)}{\partial \theta_{10}^{z_{10}} \partial \theta_{01}^{z_{01}} \partial \theta_{11}^{z_{11}}} &= \frac{(c + z_{10} -1)!}{(c - 1)!}(1 - \theta_{01} - \theta_{11})^{1 / k - z_{01} - z_{11}}\\
		&\qquad \times \sum_{j=0}^{z_{01} + z_{11}} \frac{\theta_{10}^j}{(1 - \mathbf{1}^\prime \btheta)^{c + z_{10} + j}} \binom{z_{01} + z_{11}}{j} \frac{(c + z_{10} + j - 1)!}{(c + z_{10} - 1)!}\\
		&\qquad \times \prod_{i=j}^{z_{01} + z_{11}} \frac{(c + z_{10} + i - 1 / k)}{(c + z_{10} + z_{01} + z_{11} - 1 / k)}\\
		&\qquad - \frac{\partial^{z_{10} + z_{01} + z_{11}} g_{00}(\btheta)}{\partial \theta_{10}^{z_{10}} \partial \theta_{01}^{z_{01}} \partial \theta_{11}^{z_{11}}};
		\end{aligned}$
		
		\vspace{4pt}
		\item[(iii)]
		$\begin{aligned}[t]
		\frac{\partial^{z_{10} + z_{01} + z_{11}} g_{01}(\btheta)}{\partial \theta_{10}^{z_{10}} \partial \theta_{01}^{z_{01}} \partial \theta_{11}^{z_{11}}} &= \frac{(c + z_{01} -1)!}{(c - 1)!}(1 - \theta_{10} - \theta_{11})^{1 / k - z_{10} - z_{11}}\\
		&\qquad \times \sum_{j=0}^{z_{10} + z_{11}} \frac{\theta_{01}^j}{(1 - \mathbf{1}^\prime \btheta)^{c + z_{01} + j}} \binom{z_{10} + z_{11}}{j} \frac{(c + z_{01} + j - 1)!}{(c + z_{01} - 1)!}\\
		&\qquad \times \prod_{i=j}^{z_{10} + z_{11}} \frac{(c + z_{01} + i - 1 / k)}{(c + z_{10} + z_{01} + z_{11} - 1 / k)}\\
		&\qquad - \frac{\partial^{z_{10} + z_{01} + z_{11}} g_{00}(\btheta)}{\partial \theta_{10}^{z_{10}} \partial \theta_{01}^{z_{01}} \partial \theta_{11}^{z_{11}}}.
		\end{aligned}$
		
		\end{description}
\end{lemma}

\begin{proof}
	Let $\xi = 1 / k$. For derivatives of $g_{00}(\btheta) = (1 - \mathbf{1}^\prime \btheta)^{\xi - c}$, we can use the fact that the function is symmetric in $\theta_{10}, \theta_{01},$ and $\theta_{11}$ and that the partial derivatives can be computed in any order to show the first part iteratively. This is done identically as in the proof of Corollary \ref{cor:ub1}.
	
	For $\displaystyle{g_{10}(\btheta) = \frac{(1 - \theta_{01} - \theta_{10})^{\xi}}{\theta_{00}^{c}} - g_{00}(\btheta)}$, we need only find the partial derivative of the first term. Note that this term is symmetric in $\theta_{01}$ and $\theta_{11}$ so that the problem is equivalent to finding $\displaystyle{\frac{\partial^{z_{10} + r} b(\theta_{10}, \gamma)}{\partial \theta_{10}^{z_{10}} \partial \gamma^r}}$, where $\displaystyle{b(\theta_{10}, \gamma) = \frac{(1 - \gamma)^{\xi}}{(1 - \theta_{10} - \gamma)^{c}}}$, which we do by induction.
	
	For the base case, note that for $(z_{10}, \gamma) = (0, 0)$ the result is straightforward. Assume then the result for $(z_{10}, \gamma) = (z, 0)$, so that \[\frac{\partial^z b(\theta_{10}, \gamma)}{\partial \theta_{10}^z} = \frac{(c + z - 1)!}{(c - 1)!}\frac{(1 - \gamma)^{\xi}}{(1 - \theta_{10} - \gamma)^{c + z}}.\] Differentiating with respect to $\theta_{10}$ yields
	\[\frac{\partial^{z + 1} b(\theta_{10}, \gamma)}{\partial \theta_{10}^{z + 1}} = \frac{(c + z  + 1 - 1)!}{(c - 1)!}\frac{(1 - \gamma)^{\xi}}{(1 - \theta_{10} - \gamma)^{c + z + 1}}, \] so the result holds here as well.
	
	Assume now that the result holds for $(z_{10}, \gamma) = (z_{10}, r)$ so that
		\begin{align*}
		\frac{\partial^{z_{10} + r} b(\theta_{10}, \gamma)}{\partial \theta_{10}^{z_{10}}\partial \gamma^r} &=  \frac{(c + z_{10} -1)!}{(c - 1)!}(1 - \gamma)^{\xi - r}\\ &\quad \times \sum_{j=0}^{r} \frac{\theta_{10}^j}{(1 -\theta_{10} - \gamma )^{c + z_{10} + j}} \binom{r}{j} \frac{(c + z_{10} + j - 1)!}{(c + z_{10} - 1)!}\prod_{i=j}^{r} \frac{(c + z_{10} + i - \xi)}{(c +  z_{10} + r - \xi)}.\
	\end{align*}
	
	Differentiating with respect to $\gamma$ yields
		\begin{align*}
		\frac{\partial^{z_{10} + r + 1} b(\theta_{10}, \gamma)}{\partial \theta_{10}^{z_{10}}\partial \gamma^{r + 1}} &= \frac{(c + z_{10} -1)!}{(c - 1)!}\sum_{j=0}^{r} \theta_{10}^j\binom{r}{j} \frac{(c + z_{10} + j - 1)!}{(c + z_{10} - 1)!}\\
		&\qquad \times \prod_{i=j}^{r} \frac{(c + z_{10} + i - \xi)}{(c +  z_{10} + r - \xi)}\frac{\partial}{\partial \gamma}\left(\frac{(1 - \gamma)^{\xi - r}}{(1 -\theta_{10} - \gamma )^{c + z_{10} + j}}\right),
		\end{align*} with
		\begin{align*}
		\frac{\partial}{\partial \gamma}\left(\frac{(1 - \gamma)^{\xi - r}}{(1 -\theta_{10} - \gamma )^{c + z_{10} + j}}\right) &= \frac{(c + z_{10} + j + r - \xi)(1 - \gamma)^{\xi - r - 1}}{(1 - \theta_{10} - \gamma)^{c + z_{10} + j}}\\
		&\qquad+ \frac{\theta_{10}(c + z_{10} + j)(1 - \gamma)^{\xi - r - 1}}{(1 - \theta_{10} - \gamma)^{c + z_{10} + j + 1}}.
		\end{align*}
	For the coefficient of the term $\displaystyle{\frac{(1 - \gamma)^{\xi - r - 1}}{(1 - \theta_{10} - \gamma)^{c + z_{10}}}}$ this yields \[\prod_{i=0}^r \frac{(c + z_{10} + i - \xi)}{(c + z_{10} + r - \xi)}(c + z_{10} + r - \xi) = \prod_{i=0}^{r + 1} \frac{(c + z_{10} + i - \xi)}{(c + z_{10} + r + 1- \xi)}.\]
	Likewise, for the coefficients of the terms $\frac{(1 - \gamma)^{\xi - r - 1}}{(1 - \theta_{10} - \gamma)^{c + z_{10} + j}},\ j = 1, 2, \ldots, r,$ this yields

	\begin{align*}
		&\theta_{10}^j \binom{r}{j} \frac{(c + z_{10} + j - 1)!}{(c + z_{10} - 1)!}\prod_{i=j}^{r} \frac{(c + z_{10} + i - \xi)}{(c +  z_{10} + r - \xi)} \times (c + z_{10} + j + r - \xi)\\ &\qquad + \theta_{10}^{j-1} \binom{r}{j - 1} \frac{(c + z_{10} + j - 1 - 1)!}{(c + z_{10} - 1)!}\prod_{i=j -1}^{r} \frac{(c + z_{10} + i - \xi)}{(c +  z_{10} + r - \xi)} \times \theta_{10}(c + z_{10} + j - 1),
	\end{align*} which simplifies to
	\begin{align*}
		&\theta_{10}^j \binom{r + 1}{j} \frac{(c + z_{10} + j - 1)!}{(c + z_{10} - 1)!}\prod_{i=j}^{r} \frac{(c + z_{10} + i - \xi)}{(c +  z_{10} + r - \xi)}\\
		&\qquad \times\left[ \frac{(r + j - 1)(c + z_{10} + j + r - \xi)}{r + 1} + \frac{j(c + z_{10} + j - 1 - \xi)}{r + 1}\right]
	\end{align*} and finally to
	\[\theta_{10}^j \binom{r + 1}{j} \frac{(c + z_{10} + j - 1)!}{(c + z_{10} - 1)!}\prod_{i=j}^{r + 1} \frac{(c + z_{10} + i - \xi)}{(c +  z_{10} + r + 1- \xi)}.\]
	Finally, we have a coefficient for the term $\displaystyle{\frac{(1 - \gamma)^{\xi - r - 1}}{(1 - \theta_{10} - \gamma)^{c + z_{10} + r + 1}}}$
	\[\theta_{10}^r \frac{(c + z_{10} + r - 1)!}{(c + z_{10} - 1)!} \times \theta_{10} (c + z_{10} + r) = \theta_{10}^{r + 1} \frac{(c + z_{10} + r + 1 - 1)!}{(c + z_{10} - 1)!}.\] Combining yields the desired result.
	
	Since $b(\theta_{10}, \gamma)$ is a smooth function on the indicated space, the order of derivatives is immaterial and so the above completes the proof by induction. The result for $g_{01}(\btheta)$ is identical and omitted.
\end{proof}

Now, to find an unbiased estimator of $p_{00} = h_{00}(\btheta)$, we evaluate the derivative from Lemma \ref{lm:ub2} at $\btheta = \mathbf{0}$ and use Theorem \ref{mnub} to get
\begin{align*}
\hat{p}_{00} &= \frac{(c-1)!}{(c + z_{10} + z_{01} + z_{11})!} \prod_{j=0}^{z_{10} + z_{01} + z_{11}} \frac{(c + j - \xi)}{(c + z_{10} + z_{01} + z_{11} - \xi)} \\
&= \frac{(c + z_{10} + z_{01} + z_{11})}{(c + z_{10} + z_{01} + z_{11} - \xi)} \prod_{j=0}^{z_{10} + z_{01} + z_{11}} \frac{(c + j - \xi)}{(c + j)} \\
&=\frac{1}{\left(1 - \frac{\xi}{(c + z_{10} + z_{01} + z_{11})}\right)}\prod_{j=0}^{z_{10} + z_{01} + z_{11}}\left(1 - \frac{\xi}{(c + j)}\right).
\end{align*}

The proofs for $\hat{p}_{10}$ and $\hat{p}_{01}$ are nearly identical and are omitted here.

\subsection{Proof of Theorem \ref{thm:np2}}
Let $\mathbf{h}^*(\btheta)$ be the first three components of $\mathbf{h}$ and $\widehat{\mathbf{h}^*(\btheta)}$ an unbiased estimator under $\mathcal{S}$. Now, there exist values of $\btheta \in \boldsymbol{\Psi}_{\btheta}$ such that $\mathbf{1}^\prime \mathbf{h}^*(\btheta) > 1$ (for example, $\btheta = (.45, .45, .05)$). However, $\mathbf{h}^*(\btheta)$ is analytic on all values $\boldsymbol{\Psi}_{\btheta}$, and so, by the uniqueness of the multivariate Taylor expansion, $\mathrm{E}_\mathcal{S}(\widehat{\mathbf{h}^*(\btheta)}) = \mathbf{h}^*(\btheta)$ for all $\btheta \in \boldsymbol{\Psi}_{\btheta}$, even those outside of $\boldsymbol{\Psi}_Z$. Then, if $\widehat{\mathbf{h}^*(\btheta)}$ were a proper estimator under $\mathcal{S}$, we would have $\mathrm{E}_{\mathcal{S}}(\mathbf{1}^\prime \widehat{\mathbf{h}^*(\btheta)}) \leq 1$ for all $\btheta$, so that $\widehat{\mathbf{h}^*(\btheta)}$ is not unbiased.

\subsection{Proof of Theorem \ref{thm:ident}}
Since the Multinomial is a full rank exponential family (for four class data in three dimensions), it is sufficient to
show that there exists a one-to-one mapping from $\mathbf{p}$ to
$\boldsymbol{\eta}$. Since, From Lemma \ref{lm:h2}, $\btheta$ is a one-to-one function of
$\p$, we need only show that such a mapping exists from $\btheta$ to
$\bEta$. However, from (\ref{eq:eta1}) the correspondence is one-to-one if and only if $\boldsymbol{\Phi}$ is non-singular, which is equivalent to the given condition.

\subsection{Proof of Corollary \ref{cor:ident}}
Let $\nu_1 = \pi_0^{(1)} + \pi_1^{(1)} - 1$ and $\nu_2 = \pi_0^{(2)} + \pi_1^{(2)} - 1$. Then, $\phi$ from Theorem \ref{thm:ident} reduces to
\[\arraycolsep=5pt \left|\begin{array}{ccc} \nu_1 \pi_{0}^{(2)} & -\nu_2 (1 - \pi_{0}^{(1)}) & \nu_1\pi_{0}^{(2)} - \nu_2\pi_1^{(1)} \\
-\nu_1 (1 - \pi_{0}^{(2)}) & \nu_2 \pi_{0}^{(1)} & \nu_2\pi_{0}^{(1)} - \nu_1\pi_1^{(2)} \\
\nu_1 (1 - \pi_{0}^{(2)}) & \nu_2 (1 - \pi_{0}^{(1)}) & \nu_1(1 - \pi_{0}^{(2)}) + \nu_2\pi_1^{(1)}\end{array}\right| = (\nu_1 \nu_2)^2.\] Then, $\phi \neq 0$ if and only if both $\nu_1 \neq 0$ and $\nu_2 \neq 0$.

\subsection{Proof of Theorem \ref{thm:np2mc}}
From (\ref{eq:eta1}) we have $\btheta = g(\bEta)= \boldsymbol{\Phi}^{-1} \bEta + \boldsymbol{\Phi}^{-1} \boldsymbol{\pi}_{00}$, so that $\btheta$ can be achieved by a full rank affine transformation of $\bEta$. But, this implies that the function $\mathbf{p} = h(g(\bEta))$ is analytic at the same points $\btheta \notin \boldsymbol{\Psi}_Z$ as noted in the proof of Theorem \ref{thm:np2}. The proof then follows exactly as in the previous theorem.

\section*{Acknowledgements}
The authors would like to thank the Associate Editor and two anonymous referees whose input greatly improved the presentation of this paper.




\end{document}